\documentclass[11pt]{amsart}
\usepackage{amsfonts, amsmath, amssymb, color}
\usepackage{graphicx}
\usepackage{amsthm}
\usepackage{bbm}
\usepackage{booktabs}
\usepackage{float}
\usepackage{hhline}
\usepackage{lipsum}
\usepackage{lscape}
\usepackage{subfig}
\usepackage{textcomp}
\usepackage{tikz}
\usetikzlibrary{decorations.pathmorphing,shapes,arrows,positioning}
\usepackage{xcolor}
\usepackage{young}
\usepackage{youngtab}
\usepackage{ytableau}
\usepackage{tikz}
\usepackage{float}
\usepackage{scalefnt}
\usetikzlibrary{calc}

\allowdisplaybreaks[4]
\theoremstyle{plain}
\newtheorem{thm}{Theorem}[section]
\newtheorem{lem}[thm]{Lemma}

\newtheorem{cor}[thm]{Corollary}
\newtheorem{rem}[thm]{Remark}

\theoremstyle{definition}

\newtheorem{exmp}[thm]{Example}
\newcommand{\rmnum}[1]{\romannumeral #1}
\newcommand{\Rmnum}[1]{\expandafter\@slowromancap\romannumeral #1@}

\numberwithin{equation}{section} \errorcontextlines=0
\begin{document}
\title{Tighter parameterized monogamy relations} 
\author{Yue Cao}
\address{School of Mathematics, South China University of Technology,
Guangzhou, Guangdong 510640, China}
\email{434406296@qq.com}
\author{Naihuan Jing*}
\address{Department of Mathematics, North Carolina State University, Raleigh, NC 27695, USA}
\email{jing@ncsu.edu}
\author{Kailash Misra}
\address{Department of Mathematics, North Carolina State University, Raleigh, NC 27695, USA}
\email{misra@ncsu.edu}
\author{Yiling Wang}
\address{Department of Mathematics, North Carolina State University, Raleigh, NC 27695, USA}
\email{ywang327@ncsu.edu}
\subjclass[2010]{Primary: 85; Secondary:}
\thanks{$*$Corresponding author: jing@ncsu.edu}

\begin{abstract}
We seek a systematic tightening method to represent the monogamy relation for some measure in multipartite quantum systems. By
introducing a family of parametrized bounds, we obtain tighter lowering bounds for the monogamy relation compared with the most recently discovered relations. We provide detailed examples to illustrate why our bounds are better.
\end{abstract}

\maketitle

\section{\textbf{Introduction}}
Quantum entanglement \cite{HHHH}
holds great significance in the realm of quantum information processing. It represents a special and characteristic state of a quantum system where two or more quantum particles are correlated or interdependent no matter how far they are. It is this special correlation that distinguishes a quantum system from a class system and more importantly, it is responsible for quantum supremacy and separating classical mechanics from quantum theory.

Quantum entanglement as an important quantum correlation describes the bounding
phenomenon of various sub-particles in ways that each sub-particle or state cannot be independently determined from the others. One significant phenomenon is the monogamy relation \cite{ASI}. The monogamy relation has important applications in quantum key distribution \cite{P}, foundations of quantum mechanics \cite{T,S}, quantum communications
\cite{B,CB,GR}, quantum computing \cite{EJ,NC}
etc. It underscores the idea that entanglement is not distributable without limitations,
leading to important implications for the manipulation and transmission of quantum information within multipartite quantum systems
\cite{P,AMG}.

Monogamy inequality was first discovered by Coffman, Kundu, and Wootters \cite{CKW} for concurrence in three-qubit systems:
$
C^{2}(\rho_{A|BC})\geq C^{2}(\rho_{AB})+C^{2}(\rho_{AC}),
$
where $\rho_{AB}$ and $\rho_{AC}$ are the reduced density matrices of $\rho_{ABC}$. It was quickly generalized to a multipartite system by Osborne and Verstraete \cite{OV}:
\begin{equation}\label{e:1-mono1}
C^{2}(\rho_{A|B_{1}\cdots B_{N-1}})\geq \sum_{i=1}^{N-1} C^{2}(\rho_{AB_{i}}).
\end{equation}

Monogamy has been studied for many different situations \cite{G, JHC,JLLF,ZF,KPSS, KDS,OF}. The monogamy relations
can help us to further understand the distribution of entanglement in multipartite
systems. Since it represents certain constrain among different subsystems, it is natural to look for tighter
monogamy relations and this may provide
a better characterization of the distribution of quantum correlations. This may lead to tighter
security bounds in quantum cryptography \cite{GRTZ}.

One common method to study monogamy relations
is to derive the consequential monogamy relations by bounding the binomial function
$(1+t)^{x}$ using various smart estimates, see
for example \cite{JFQ,ZJZ1,ZJZ2,ZLJM}.
By drastically improving the bound, some of us have recently found a unified and improved method to estimate its lower bounds and derive tighter monogamy relations \cite{CJW}.

In \cite{GYG1} the authors have formulated the transitivity property for the monogamy relation. Suppose a monogamy relation for certain measure $\mathcal{E}$ holds for the infinimum $\alpha_c$:
\begin{align}\label{e:1-mono2}
\mathcal{E}^{\alpha_{c}}(\rho_{A|B_{1}\cdots B_{N-1}})\geq \sum_{i=1}^{N-1} \mathcal{E}^{\alpha_{c}}(\rho_{AB_{i}}).
\end{align}
then the monogamy relation also holds for any $\alpha\geq \alpha_c$. Here $\rho_{AB_{i}}=\operatorname{Tr}_{B_{1}\cdots B_{i-1}B_{i+1}\cdots B_{N-1}}(\rho_{A B_{1}\cdots B_{N-1}})$  is the reduced density matrix.
Applying the monogamy relation in \cite{GYG1} to quantum correlations like squared convex-roof extended negativity, entanglement of formation, and concurrence one can get tighter monogamy inequalities.
Recently, the authors \cite{TZJF,LSF} rigorously proved that their monogamy is tighter than those given before. This prompts us to seek further improvement for the
the monogamy relations.

In this study, we propose a new method by introducing parametrized bounds. Our idea is to give a family of tighter monogamy relations in a unified manner. We know that all these work for the quantum correlation measure such as the concurrence or convex-roof extended
negativity. In this case, we
obtain better monogamy relations than \cite{TZJF,LSF} as well as those from \cite{JLLF,ZF,YCFW,JF1}. Furthermore, We can choose appropriate parameter to get tighter monogamy relations.

This paper is organized as follows. In Section
\ref{s:prelim} we prepare the necessary mathematical tools
to deal with the approximation.
In Section \ref{s:monogamy} we give tighter monogamy
relations based on the mathematical results
in Section \ref{s:prelim}. In the last section, we give
three examples to show why our new bounds are stronger than some of the recently found sharper bounds.

\section{\textbf{Preliminaries about several important inequalities }}\label{s:prelim}
In the study of the $\alpha$th-monogamy relations, we usually consider how to bound the binomial function $(1+t)^x$ over an interval.

We first consider some basic inequalities.
For $0\leq t\leq k, 0<k\leq1, x\geq 2$, the authors have shown in \cite{TZJF,LSF} that
\begin{equation}\label{e:2-1}
\begin{aligned}
(1+t)^{x}\geq 1+\left(\frac{(1+k)^{x}-1}{k^{x}}+k^{x}-t^{x}\right)t^{x}\geq 1+(2^{x}-t^{x})t^{x}.
\end{aligned}
\end{equation}
Furthermore, they also proved that for $ t\geq k\geq1, x\geq 2$,
\begin{equation}\label{e:2-2}
\begin{aligned}
(1+t)^{x}\geq t^{x}+(1+k)^{x}-k^{x}+k^{-x}-t^{-x} \geq t^{x}+2^{x}-t^{-x}.
\end{aligned}
\end{equation}

We now introduce a parameter to get a family of
stronger inequalities than \eqref{e:2-1}--\eqref{e:2-2}, which will prepare the mathematical background for strengthening the monogamy relations of the multipartite quantum.

\begin{lem}\label{e:2-lem1} For any parameter $m\geq 0$, suppose $x\geq 1+{\rm log}_{2}(m+2)\geq2$.\\
$(\rmnum{1})$ Let $0< t\leq k\leq1$, then
\begin{equation}\label{e:2-3}
\begin{aligned}
(1+t)^{x}\geq 1+\left(\frac{(1+k)^{x}-1}{k^{x}}+k^{x}-t^{x}\right)t^{x}+mx\left(\frac{1}{t}-\frac{1}{k}\right)t^{x}.
\end{aligned}
\end{equation}
$(\rmnum{2})$ Let $ t\geq k\geq1$, then we have
\begin{equation}\label{e:2-4}
\begin{aligned}
(1+t)^{x}\geq t^{x}+(1+k)^{x}-k^{x}+k^{-x}-t^{-x}+mx(t-k).
\end{aligned}
\end{equation}
\end{lem}
 \begin{proof} Fix $m\geq 0$,
consider the following function:
$$h(x,y)=(1+y)^{x}-y^{x}+y^{-x}-mxy,$$
defined on $(x,y)\in [1+{\rm log}_{2}(m+2),+\infty)\times[1,+\infty).$
Note that
\begin{align*}
\frac{\partial h}{\partial y}(x,y)=x(1+y)^{x-1}-xy^{x-1}-xy^{-x-1}-mx
\end{align*}
and
\begin{align*}
\frac{\partial^{2} h}{\partial y^{2}}(x,y)=x(x-1)y^{x-2}\left((1+\frac{1}{y})^{x-2}-1\right)+x(x+1)y^{-x-2}.
\end{align*}
So $\frac{\partial^{2} h}{\partial y^{2}}\geq0$ for
$x\geq 2$. This means that for fixed $x\geq 2$, the function $\frac{\partial
h}{\partial y}$ is increasing as a function of $y$.
Therefore for $y\geq 1$
\begin{equation}\label{e:der1}
\frac{\partial h}{\partial y}(x,y)\geq\frac{\partial h}{\partial y}(x,1)=x(2^{x-1}-2-m).
\end{equation}
The right-hand side of \eqref{e:der1} is nonnegative when $x\geq 1+\log_{2}(m+2)$. Note that $\log_2(m+2)\geq 1$ for $m\geq 0$. Therefore $h(x,y)$ is increasing with respect to $y$ for $x\geq 1+\log_{2}(m+2)$. Subsequently,

$(\rmnum{1})$ for $0<t\leq k\leq1$, we have
$h(x,\frac{1}{t})\geq h(x,\frac{1}{k})\geq h(x,1)$, which is \eqref{e:2-3}.

$(\rmnum{2})$ for $ t\geq k\geq1$, we have $h(x,t)\geq h(x,k)\geq h(x,1)$, which immeditately implies \eqref{e:2-4}.
\end{proof}

The following result is the special case of $k=1$ in Lemma \ref{e:2-lem1}.
\begin{cor}\label{e:2-cor1} Let $m\geq 0$ be a parameter and $x\geq 1+{\rm log}_{2}(m+2)\geq2$.\\
$(\rmnum{1})$ If $0< t\leq1$, then
\begin{equation}\label{e:2-5}
\begin{aligned}
(1+t)^{x}\geq 1+(2^{x}-t^{x})t^{x}+mx(\frac{1}{t}-1)t^{x}.
\end{aligned}
\end{equation}
$(\rmnum{2})$ If $ t\geq1$, then we have
\begin{equation}\label{e:2-6}
\begin{aligned}
(1+t)^{x}\geq t^{x}+2^{x}-t^{-x}+mx(t-1).
\end{aligned}
\end{equation}
\end{cor}

We claim that our new inequalities are stronger than previous studies. This is substantiated by comparing with the recently available strong inequalities. There are two cases to consider.

{\bf Case 1:} For $0< t\leq k\leq1$, Lemma \ref{e:2-lem1} $(i)$  actually shows the following
\begin{align*}
 (1+t)^{x}&\geq 1+\left(\frac{(1+k)^{x}-1}{k^{x}}+k^{x}-t^{x}\right)t^{x}+mx\left(\frac{1}{t}-\frac{1}{k}\right)t^{x}\\
 &\geq 1+\left(\frac{(1+k)^{x}-1}{k^{x}}+k^{x}-t^{x}\right)t^{x}\quad(\text{That of \cite{LSF}})\\
 &\geq 1+\left(\frac{(1+k)^{x}-1}{k^{x}}\right)t^{x}\quad(\text{The inequality in \cite{YCFW}})\\
 &\geq 1+(2^{x}-1)t^{x}\quad(\text{The inequality in \cite{JLLF}})\\
 &\geq 1+xt^{x}\quad(\text{The inequality in \cite{JF1}})\\
 &\geq 1+t^{x}\quad(\text{The inequality in \cite{ZF}}).
\end{align*}
{\bf Case 2:} For $ t\geq k\geq1$, Lemma \ref{e:2-lem1} $(ii)$ gives that
\begin{align*}
(1+t)^{x}&\geq t^{x}+(1+k)^{x}-k^{x}+k^{-x}-t^{-x}+mx(t-k)\\
&\geq t^{x}+(1+k)^{x}-k^{x}+k^{-x}-t^{-x}\quad(\text{by \eqref{e:2-2}})\\
&\geq t^{x}+(1+k)^{x}-k^{x}\quad(\text{The inequality in \cite{GYG1}})\\
&\geq t^{x}+2^{x}-1\quad(\text{The inequality in \cite{JLLF}})\\
&\geq t^{x}+1\quad(\text{The inequality in \cite{ZF}})
\end{align*}
for $x\geq 1+{\rm log}_{2}(m+2)$ and arbitrary $ m\geq0$.

In fact, the inequalities in \cite{TZJF,LSF} are the special cases of Lemma \ref{e:2-lem1} and Corollary \ref{e:2-cor1} for $m=0$. Our parameterized bounds (via the parameter $m$) are stronger in all levels in the sense that we divide the region into smaller ones on which our parameterized bounds are tighter than the previously available bounds in \cite{JLLF,ZF,GYG1,TZJF,LSF,YCFW,JF1}.

\section{\textbf{Tighter $\alpha$th power of entanglement measures }}\label{s:monogamy}

Let $\rho=\rho_{AB_{1}\cdots B_{N-1}}$ be a quantum state over the Hilbert space $\mathcal{H}_{A}\bigotimes $ $ \mathcal{H}_{B_{1}} $ $
\bigotimes $ $ \cdots$ $ \bigotimes
\mathcal{H}_{B_{N-1}}$, and $\mathcal{E}$
a bipartite entanglement measure of quantum correlation. The $\gamma$th-monogamy $(\gamma> 0)$ of the measure $\mathcal{E}$ is defined as
\cite[(1)]{GYG1}
\begin{align}\label{e:3-mono1}
\mathcal{E}^{\gamma}(\rho_{A|B_{1}\cdots B_{N-1}})\geq \sum_{i=1}^{N-1} \mathcal{E}^{\gamma}(\rho_{AB_{i}}).
\end{align}
where $\rho_{AB_{i}}$ is the reduced density matrix. The exponent $\gamma$ depends on the infimum of all indices satisfying monogamy
relation \eqref{e:3-mono1} of measure $\mathcal{E}$ (eg. If $\mathcal{E}=C$, then $\gamma=2$ ).
The $\gamma$th-monogamy relations are applicable to some quantum correlation measures, such as concurrence \cite{U,RBCGM,AF}, negativity \cite{VW}, logarithmic negativity \cite{P1}, entanglement of formation \cite{JLLF,ZF,JF}, etc.
We can study the generalized monogamy relations in a uniform manner as follows.

From now on, if the state $\rho$ is clear from the context, we simply write $\mathcal{E}(\rho_{AB_{i}})=\mathcal{E}_{AB_{i}}$,
$\mathcal{E}(\rho_{A|B_{1}B_{2}\cdots B_{N-1}})=\mathcal{E}_{A|B_{1}B_{2}\cdots B_{N-1}}$.

\subsection{\textbf{Tighter $\alpha$th-monogamy relation based on {\bf Case 1}}}

We will first give some tighter $\alpha$th power inequalities based on Lemma \ref{e:2-lem1} $(\rmnum{1})$ and Corollary \ref{e:2-cor1} $(\rmnum{1})$.

\begin{thm}\label{e:3.1-thm1}
Let $\mathcal{E}$ be a bipartite quantum measure satisfying the $\gamma$th-monogamy relation \eqref{e:3-mono1} for the tripartite state
$\rho_{ABC}$.
If $\mathcal{E}^{\gamma}_{AB}\geq \mathcal{E}^{\gamma}_{AC}>0$ $(\gamma>0)$, then one has that
\begin{align}\label{e:3.1-1}
\mathcal{E}_{A|BC}^{\alpha}\geq\mathcal{E}_{AB}^{\alpha}+\left[\left(2^{\frac{\alpha}{\gamma}}-\frac{\mathcal{E}^{\alpha}_{AC}}{\mathcal{E}^{\alpha}_{AB}}\right)+\frac{m\alpha}{\gamma}\left(\frac{\mathcal{E}_{AB}^{\gamma}}{\mathcal{E}_{AC}^{\gamma}}-1\right)\right]\mathcal{E}_{AC}^{\alpha}
\end{align}
for any $\alpha\geq \left[1+{\rm log}_{2}(m+2)\right]\gamma\geq 2\gamma>0$, $m\geq0$.
\end{thm}
\begin{proof} It follows from \eqref{e:3-mono1} and Corollary \ref{e:2-cor1} $(\rmnum{1})$ that
\begin{align*}
\mathcal{E}_{A|BC}^{\alpha}&=(\mathcal{E}_{A|BC}^{\gamma})^{\frac{\alpha}{\gamma}}\geq
(\mathcal{E}^{\gamma}_{AB}+\mathcal{E}^{\gamma}_{AC})^{\frac{\alpha}{\gamma}}=\mathcal{E}_{AB}^{\alpha}\left(1+\frac{\mathcal{E}^{\gamma}_{AC}}{\mathcal{E}^{\gamma}_{AB}}\right)^{\frac{\alpha}{\gamma}}\\
&\geq\mathcal{E}_{AB}^{\alpha}\left[1+\left(2^{\frac{\alpha}{\gamma}}-\left(\frac{\mathcal{E}^{\gamma}_{AC}}{\mathcal{E}^{\gamma}_{AB}}\right)^{\frac{\alpha}{\gamma}}\right)\left(\frac{\mathcal{E}^{\gamma}_{AC}}{\mathcal{E}^{\gamma}_{AB}}\right)^{\frac{\alpha}{\gamma}}+\frac{m\alpha}{\gamma}\left(\frac{\mathcal{E}^{\gamma}_{AB}}{\mathcal{E}^{\gamma}_{AC}}-1\right)\left(\frac{\mathcal{E}^{\gamma}_{AC}}{\mathcal{E}^{\gamma}_{AB}}\right)^{\frac{\alpha}{\gamma}}\right]\\
&=\mathcal{E}_{AB}^{\alpha}+\left[2^{\frac{\alpha}{\gamma}}-\frac{\mathcal{E}^{\alpha}_{AC}}{\mathcal{E}^{\alpha}_{AB}}+\frac{m\alpha}{\gamma}\left(\frac{\mathcal{E}^{\gamma}_{AB}}{\mathcal{E}^{\gamma}_{AC}}-1\right)\right]\mathcal{E}_{AC}^{\alpha}.
\end{align*}
\end{proof}

Theorem \ref{e:3.1-thm1} also says that the lower bound can be strengthened according to the
value of $\alpha$. As $\alpha$ increases, the monogamy relation becomes tighter. In fact, for the $\alpha$th-monogamy relation, we can choose an appropriate value of $m$ to optimize the
lower bounds. We can simply choose $m=\lfloor 2^{\frac{\alpha}{\gamma}-1}-2\rfloor$, then the $\alpha$-th monogamy is clearly better than
the monogamy relation \eqref{e:3.1-1} with $m$ replaced by $m-1$.

If we know \'a priori some information about the subsystems such as the ratio $\mathcal{E}^{\gamma}_{AB}/\mathcal{E}^{\gamma}_{AC}$, then we can improve the monogamy relation.

\begin{thm}\label{e:3.1-thm2}
Let $\mathcal{E}$ be a bipartite quantum measure satisfying the $\gamma$th-monogamy relation \eqref{e:3-mono1} for the tripartite state
$\rho_{ABC}$. If $k\mathcal{E}^{\gamma}_{AB}\geq \mathcal{E}^{\gamma}_{AC}>0$ $(\gamma>0)$ and $0<k\leq1$, then one has that
\begin{align}\label{e:3.1-2}
\mathcal{E}_{A|BC}^{\alpha}\geq\mathcal{E}_{AB}^{\alpha}+\left[\frac{(1+k)^{\frac{\alpha}{\gamma}}-1}{k^{\frac{\alpha}{\gamma}}}+k^{\frac{\alpha}{\gamma}}-\frac{\mathcal{E}^{\alpha}_{AC}}{\mathcal{E}^{\alpha}_{AB}}+\frac{m\alpha}{\gamma}\left(\frac{\mathcal{E}_{AB}^{\gamma}}{\mathcal{E}_{AC}^{\gamma}}-\frac{1}{k}\right)\right]\mathcal{E}_{AC}^{\alpha}
\end{align}
for any $\alpha\geq \left[1+{\rm log}_{2}(m+2)\right]\gamma\geq 2\gamma>0$, $m\geq0$.
\end{thm}
\begin{proof} Similar to the proof of Theorem \ref{e:3.1-thm1}, it follows from the $\gamma$th-monogamy relation \eqref{e:3-mono1} and Lemma \ref{e:2-lem1} $(\rmnum{1})$.
\end{proof}

With this improvement, we discuss how to generalize the monogamy relations to multipartite quantum system. Clearly, we can partition the subsystems into two groups according to the ratio
of complementary subsystems being greater than 1 or less than 1 and apply the result of bipartite monogamy relations repeatedly to obtain the overall monogamy relation.

\begin{thm}\label{e:3.1-thm3} Let $\mathcal{E}$ be a bipartite quantum measure satisfying the $\gamma$th-monogamy relation \eqref{e:3-mono1} for $\gamma> 0$ and $\rho_{AB_{1}\cdots B_{N-1}}$ any $N$-qubit quantum state. Suppose $0<k\leq1$. If $k\mathcal{E}_{AB_{i}}^{\gamma}\geq\mathcal{E}_{A|B_{i+1}\cdots B_{N-1}}^{\gamma}$ for $i=1,\cdots,n$ and $\mathcal{E}_{AB_{j}}^{\gamma}\leq k\mathcal{E}_{A|B_{j+1}\cdots B_{N-1}}^{\gamma}$  for $j=n+1,\cdots,N-2$~$(1\leq n\leq N-3, N\geq4)$. Then for $\alpha\geq \left[1+{\rm log}_{2}(m+2)\right]\gamma\geq 2\gamma>0$, $m\geq0$, we have
\begin{equation}
\begin{aligned}\label{e:3.1-3}
&\mathcal{E}_{A|B_{1}\cdots B_{N-1}}^{\alpha} \geq \mathcal{E}^{\alpha}_{AB_{1}}+ \sum_{i=1}^{n-1}\left(\prod_{j=1}^{i}(M_{j}+T_{j})\right)\mathcal{E}^{\alpha}_{AB_{i+1}}+\\
&\left(\prod_{i=1}^{n}(M_{i}+T_{i})\right)\left(\mathcal{E}^{\alpha}_{AB_{N-1}}+\sum^{N-2}_{j=n+1}(Q_{j}+P_{j})\mathcal{E}^{\alpha}_{AB_{j}}
\right),
\end{aligned}
\end{equation}
where $M=\frac{(1+k)^{\frac{\alpha}{\gamma}}-1}{k^{\frac{\alpha}{\gamma}}}+k^{\frac{\alpha}{\gamma}}$, and for $i=1,2,\cdots,n$, $$M_{i}=M-\frac{\mathcal{E}^{\alpha}_{A|B_{i+1}\cdots B_{N-1}}}{\mathcal{E}^{\alpha}_{AB_{i}}}, T_{i}=\frac{m\alpha}{\gamma}\left(\frac{\mathcal{E}^{\gamma}_{AB_{i}}}{\mathcal{E}^{\gamma}_{A|B_{i+1}\cdots B_{N-1}}}-\frac{1}{k}\right),$$
and for $j=n+1,n+2,\cdots,N-2$,
$$Q_{j}=M-\frac{\mathcal{E}^{\alpha}_{AB_{j}}}{\mathcal{E}^{\alpha}_{A|B_{j+1}\cdots B_{N-1}}}, P_{j}=\frac{m\alpha}{\gamma}\left(\frac{\mathcal{E}^{\gamma}_{A|B_{j+1}\cdots B_{N-1}}}{\mathcal{E}^{\gamma}_{AB_{j}}}-\frac{1}{k}\right).$$
\end{thm}
\begin{proof} Since $k\mathcal{E}_{AB_{i}}^{\gamma}\geq\mathcal{E}_{A|B_{i+1}\cdots B_{N-1}}^{\gamma}$, $i=1,\cdots,n$, it follows from Theorem \ref{e:3.1-thm2} that
\begin{align*}
\mathcal{E}_{A|B_{1}\cdots B_{N-1}}^{\alpha}&=(\mathcal{E}_{A|B_{1}\cdots B_{N-1}}^{\gamma})^{\frac{\alpha}{\gamma}}\\
&\geq(\mathcal{E}_{AB_{1}}^{\gamma}+\mathcal{E}_{A|B_{2}\cdots B_{N-1}}^{\gamma})^{\frac{\alpha}{\gamma}}\\
&\geq \mathcal{E}_{AB_{1}}^{\alpha}+(M_{1}+T_{1})\mathcal{E}_{A|B_{2}\cdots B_{N-1}}^{\alpha}\\
&\geq \mathcal{E}_{AB_{1}}^{\alpha}+(M_{1}+T_{1})\mathcal{E}_{A|B_{2}}^{\alpha}+(M_{1}+T_{1})(M_{2}+T_{2})\mathcal{E}_{A|B_{3}\cdots B_{N-1}}^{\alpha}\\
&\geq\cdots\\
&\geq \mathcal{E}_{AB_{1}}^{\alpha}+(M_{1}+T_{1})\mathcal{E}_{A|B_{2}}^{\alpha}+\cdots+\left[(M_{1}+T_{1})\cdots(M_{n-1}+T_{n-1})\right]\mathcal{E}_{AB_{n}}^{\alpha}\\
&\quad +\left[(M_{1}+T_{1})\cdots(M_{n-1}+T_{n-1})(M_{n}+T_{n})\right]\mathcal{E}_{A|B_{n+1}\cdots B_{N-1}}^{\alpha}\\
&= \mathcal{E}^{\alpha}_{AB_{1}}+ \sum_{i=1}^{n-1}\left(\prod_{j=1}^{i}(M_{j}+T_{j})\right)\mathcal{E}^{\alpha}_{AB_{i+1}}+
\left(\prod_{i=1}^{n}(M_{i}+T_{i})\right)\mathcal{E}_{A|B_{n+1}\cdots B_{N-1}}^{\alpha}.
\end{align*}
As $\mathcal{E}_{AB_{j}}^{\gamma}\leq k\mathcal{E}_{A|B_{j+1}\cdots B_{N-1}}^{\gamma}$, $j=n+1,\cdots,N-2$, using Theorem \ref{e:3.1-thm2}, we have
\begin{align*}
&\mathcal{E}_{A|B_{n+1}\cdots B_{N-1}}^{\alpha}\geq(\mathcal{E}_{AB_{n+1}}^{\gamma}+\mathcal{E}_{A|B_{n+2}\cdots B_{N-1}}^{\gamma})^{\frac{\alpha}{\gamma}}\\
&\geq \mathcal{E}_{A|B_{n+2}\cdots B_{N-1}}^{\alpha}+(Q_{n+1}+P_{n+1})\mathcal{E}_{AB_{n+1}}^{\alpha}\\
&\geq \mathcal{E}_{A|B_{n+3}\cdots B_{N-1}}^{\alpha}+(Q_{n+2}+P_{n+2})\mathcal{E}_{AB_{n+2}}^{\alpha}+(Q_{n+1}+P_{n+1})\mathcal{E}_{AB_{n+1}}^{\alpha}\\
&\geq\cdots\\
&\geq\mathcal{E}_{AB_{N-1}}^{\alpha}+(Q_{N-2}+P_{N-2})\mathcal{E}_{AB_{N-2}}^{\alpha}+\cdots+(Q_{n+1}+P_{n+1})\mathcal{E}_{AB_{n+1}}^{\alpha}\\
&=\left(\mathcal{E}^{\alpha}_{AB_{N-1}}+\sum^{N-2}_{j=n+1}(Q_{j}+P_{j})\mathcal{E}^{\alpha}_{AB_{j}}\right).
\end{align*}
\end{proof}

A special case of Theorem \ref{e:3.1-thm3} is
often used and recorded as follows.
\begin{cor}\label{e:3.1-thm4} Let $\mathcal{E}$ be a bipartite quantum measure satisfying the $\gamma$th-monogamy relation \eqref{e:3-mono1}
for $\gamma> 0$ and $\rho_{AB_{1}\cdots B_{N-1}}$ be any $N$-qubit quantum states. Suppose $0<k\leq1$. If $k\mathcal{E}_{AB_{i}}^{\gamma}\geq\mathcal{E}_{A|B_{i+1}\cdots B_{N-1}}^{\gamma}$ for $i=1,\cdots,N-2$. Then we have
\begin{equation}
\begin{aligned}\label{e:3.1-4}
&\mathcal{E}_{A|B_{1}\cdots B_{N-1}}^{\alpha} \geq \mathcal{E}^{\alpha}_{AB_{1}}+ \sum_{i=1}^{N-2}\left(\prod_{j=1}^{i}(M_{j}+T_{j})\right)\mathcal{E}^{\alpha}_{AB_{i+1}}
\end{aligned}
\end{equation}
for $\alpha\geq \left[1+{\rm log}_{2}(m+2)\right]\gamma\geq 2\gamma>0$, $m\geq0$, where $M_{i}, T_{i}$ $(i=1,\cdots,N-2)$ are defined in Theorem \ref{e:3.1-thm3}.
\end{cor}

\bigskip

For an arbitrary $N$-qubit state $\rho_{AB_{1}\cdots B_{N-1}}$, the concurrence \cite{U,RBCGM,AF} satisfies the $2$nd-monogamy relation \eqref{e:1-mono1}, and the square of convex-roof extended negativity \cite{KDS,LCO,LL} also satisfies $2$nd-monogamy relation:
\begin{equation}\label{e:mono-CREN}
\mathcal{N}_{c}^{2}(\rho_{A|B_{1}\cdots B_{N-1}})\geq \sum_{i=1}^{N-1} \mathcal{N}_{c}^{2}(\rho_{AB_{i}}).
\end{equation}

It is easy to compare the new monogamy relations with previous bounds.
\begin{rem}
If bipartite quantum measure $\mathcal{E}$ is the concurrence $C$ or the convex-roof extended negativity $\mathcal{N}_{c}$, then

$(\rmnum{1})$ the monogamy relations of \cite[Lem. 2, Thm. 1-4]{TZJF} are the special cases of Theorem \ref{e:3.1-thm1}-\ref{e:3.1-thm4} for $m=0, k=1$.

$(\rmnum{2})$ the monogamy relations of \cite[Lem. 2, Thm. 1-4]{LSF} are the particular cases of Theorem \ref{e:3.1-thm1}-\ref{e:3.1-thm4} at $m=0$.\\
\end{rem}
Therefore the monogamy relations from Theorem \ref{e:3.1-thm1}-\ref{e:3.1-thm4} for concurrence and convex-roof extended negativity are sharper than those given in \cite{TZJF,LSF}, and subsequently also \cite{JLLF,ZF,JF1}.
Later, we will give an example to show the comparisons.

\subsection{\textbf{Tighter $\alpha$th-monogamy relation based on {\bf Case 2}}}
We will give some tighter $\alpha$th power inequalities based on Lemma \ref{e:2-lem1} $(\rmnum{2})$ and Corollary \ref{e:2-cor1} $(\rmnum{2})$.

\begin{lem}\label{e:3-lem1} Let $p_1\geq \cdots \geq p_N>0$ be $N$ numbers such that $p_{i}\geq  p_{i+1}$ $(i=1,\cdots,N-1)$, then one has that
\begin{equation}\label{e:3.2-1}
\begin{aligned}
&\left(\sum_{i=1}^{N}p_{i}\right)^{x} =\left(p_{1}+p_{2}+\cdots+p_{N}\right)^{x}\\
&\geq p_{1}^{x}+(2^{x}-1)p_{2}^{x}+\cdots+\left[N^{x}-(N-1)^{x}\right]p_{N}^{x}+
\left[1-\left(\frac{p_{1}}{p_{2}}\right)^{-x}\right]p_{2}^{x}\\
&+\left[2^{-x}-\left(\frac{p_{1}+p_{2}}{p_{3}}\right)^{-x}\right]p_{3}^{x}+\cdots+\left[(N-1)^{-x}-\left(\frac{p_{1}+\cdots+p_{N-1}}{p_{N}}\right)^{-x}\right]p_{N}^{x}\\
&+mx\left[\left(\frac{p_{1}}{p_{2}}-1\right)p_{2}^{x}+\cdots+\left(\frac{p_{1}+p_{2}+\cdots+p_{N-1}}{p_{N}}-(N-1)\right)p_{N}^{x}\right].
\end{aligned}
\end{equation}
\end{lem}
\begin{proof} We use induction on $N$. The case of $N = 1$ is trivial.
Consider $N$ decreasing positive numbers $p_1\geq p_2\geq \ldots\geq p_N> 0$.
Setting $t=\frac{p_{1}+p_{2}+\cdots+p_{N-1}}{p_{N}}\geq N-1$, Lemma \ref{e:2-lem1} $(\rmnum{2})$ implies that
\begin{align*}
\left(\sum_{i=1}^{N}p_{i}\right)^{x} &=p_{N}^{x}\left(1+\frac{p_{1}+p_{2}+\cdots+p_{N-1}}{p_{N}}\right)^{x}\\
&\geq (p_{1}+p_{2}+\cdots+p_{N-1})^{x}+\left[N^{x}-(N-1)^{x}\right)]p_{N}^{x}\\
&+\left[(N-1)^{-x}-\left(\frac{p_{1}+p_{2}+\cdots+p_{N-1}}{p_{N}}\right)^{-x}\right]p_{N}^{x}\\
&+mx\left[\frac{p_{1}+p_{2}+\cdots+p_{N-1}}{p_{N}}-(N-1)\right]p_{N}^{x}.
\end{align*}
Thus, the inequality \eqref{e:3.2-1} can be directly obtained by the inductive hypothesis.
\end{proof}

The following result is immediate for any $N$-qubit quantum states $\rho_{AB_{1}\cdots B_{N-1}}$.

\begin{thm} \label{e:3.2-thm1}
Let $\mathcal{E}$ be a bipartite quantum measure satisfying the relation \eqref{e:3-mono1} for $\gamma> 0$ and $\rho_{AB_{1}\cdots B_{N-1}}$  any $N$-qubit quantum state. Arrange \{$\mathcal{E}_{i}=\mathcal{E}_{AB_{i'}}|i=1,\cdots,N-1\}$ in descending order.  If $\mathcal{E}^{\gamma}_{i}\geq \mathcal{E}^{\gamma}_{i+1}>0$ for $i=1,\cdots,N-2$, then
\begin{equation}\label{e:3.2-2}
\begin{aligned}
\mathcal{E}_{A|B_{1}\cdots B_{N-1}}^{\alpha}&\geq \mathcal{E}_{1}^{\alpha}+\sum_{i=2}^{N-1}\left[i^{\frac{\alpha}{\gamma}}-(i-1)^{\frac{\alpha}{\gamma}}+(i-1)^{-\frac{\alpha}{\gamma}}-\tau_{i}^{-\frac{\alpha}{\gamma}}\right]\mathcal{E}_{i}^{\alpha}\\
&+\sum_{i=2}^{N-1}m\frac{\alpha}{\gamma}\left[\tau_{i}-(i-1)\right]\mathcal{E}_{i}^{\alpha}
\end{aligned}
\end{equation}
for $\alpha\geq \left[1+{\rm log}_{2}(m+2)\right]\gamma\geq 2\gamma>0$, $m\geq0$, where $\tau_{i}=\frac{\mathcal{E}_{1}^{\gamma}+\cdots+\mathcal{E}_{i-1}^{\gamma}}{\mathcal{E}_{i}^{\gamma}}$, $i=2,\cdots,N-1$.
\end{thm}
\begin{proof} From the $\gamma$th-monogamy relation \eqref{e:3-mono1} and Lemma \ref{e:3-lem1} we have
\begin{align*}
&\mathcal{E}_{A|B_{1}\cdots B_{N-1}}^{\alpha}\geq (\mathcal{E}_{AB_{1}}^{\gamma}+\mathcal{E}_{AB_{2}}^{\gamma}+\cdots+\mathcal{E}_{AB_{N-1}}^{\gamma})^{\frac{\alpha}{\gamma}}= (\mathcal{E}_{1}^{\gamma}+\mathcal{E}_{2}^{\gamma}+\cdots+\mathcal{E}_{N-1}^{\gamma})^{\frac{\alpha}{\gamma}}\\
&\geq \mathcal{E}_{1}^{\alpha}+(2^{\frac{\alpha}{\gamma}}-1)\mathcal{E}_{2}^{\alpha}+\cdots+\left[(N-1)^{\frac{\alpha}{\gamma}}-(N-2)^{\frac{\alpha}{\gamma}}\right]\mathcal{E}_{N-1}^{\alpha}\\
&+\left[1-\left(\frac{\mathcal{E}_{1}^{\gamma}}{\mathcal{E}_{2}^{\gamma}}\right)^{-\frac{\alpha}{\gamma}}\right]\mathcal{E}_{2}^{\alpha}
+\cdots+\left[(N-2)^{-\frac{\alpha}{\gamma}}-\left(\frac{\mathcal{E}_{1}^{\gamma}+\cdots+\mathcal{E}_{N-2}^{\gamma}}{\mathcal{E}_{N-1}^{\gamma}}\right)^{-\frac{\alpha}{\gamma}}\right]\mathcal{E}_{N-1}^{\alpha}\\
&+m\frac{\alpha}{\gamma}\left(\frac{\mathcal{E}_{1}^{\gamma}}{\mathcal{E}_{2}^{\gamma}}-1\right)\mathcal{E}_{2}^{\alpha}+\cdots+m\frac{\alpha}{\gamma}\left(\frac{\mathcal{E}_{1}^{\gamma}+\cdots+\mathcal{E}_{N-2}^{\gamma}}{\mathcal{E}_{N-1}^{\gamma}}-(N-2)\right)\mathcal{E}_{N-1}^{\alpha}\\
&=\mathcal{E}_{1}^{\alpha}+\sum_{i=2}^{N-1}\left[i^{\frac{\alpha}{\gamma}}-(i-1)^{\frac{\alpha}{\gamma}}+(i-1)^{-\frac{\alpha}{\gamma}}-\tau_{i}^{-\frac{\alpha}{\gamma}}+m\frac{\alpha}{\gamma}\left(\tau_{i}-(i-1)\right)\right]\mathcal{E}_{i}^{\alpha}.
\end{align*}
\end{proof}

The following Corollary give a comparison of the monogamy relations for any quantum measure.
\begin{cor}\label{e:3.2-cor1}
Let $\mathcal{E}$ be a bipartite quantum measure satisfying the $\gamma$th-monogamy relation \eqref{e:3-mono1} for $\gamma> 0$ and $\rho_{AB_{1}\cdots B_{N-1}}$ any $N$-qubit quantum state. By permuting $B_1, \ldots, B_{N-1}$,
we assume \{$\mathcal{E}_{i}=\mathcal{E}_{AB_{i'}}|i=1,\cdots,N-1\}$ is in descending order.  If $\mathcal{E}^{\gamma}_{i}\geq \mathcal{E}^{\gamma}_{i+1}>0$ for $i=1,\cdots,N-2$, then
\begin{equation}\label{e:3.2-3}
\begin{aligned}
\mathcal{E}_{A|B_{1}\cdots B_{N-1}}^{\alpha}&\geq \mathcal{E}_{1}^{\alpha}+\sum_{i=2}^{N-1}\left[i^{\frac{\alpha}{\gamma}}-(i-1)^{\frac{\alpha}{\gamma}}+(i-1)^{-\frac{\alpha}{\gamma}}-\tau_{i}^{-\frac{\alpha}{\gamma}}\right]\mathcal{E}_{i}^{\alpha}\\
&+\sum_{i=2}^{N-1}m\frac{\alpha}{\gamma}\left[\tau_{i}-(i-1)\right]\mathcal{E}_{i}^{\alpha}\\
&\geq \mathcal{E}_{1}^{\alpha}+\sum_{i=2}^{N-1}\left[i^{\frac{\alpha}{\gamma}}-(i-1)^{\frac{\alpha}{\gamma}}+(i-1)^{-\frac{\alpha}{\gamma}}-\tau_{i}^{-\frac{\alpha}{\gamma}}\right]\mathcal{E}_{i}^{\alpha}\\
&\geq \mathcal{E}_{1}^{\alpha}+\sum_{i=2}^{N-1}\left[i^{\frac{\alpha}{\gamma}}-(i-1)^{\frac{\alpha}{\gamma}}\right]\mathcal{E}_{i}^{\alpha}\\
&\geq \mathcal{E}_{1}^{\alpha}+\sum_{i=2}^{N-1}\left(2^{\frac{\alpha}{\gamma}}-1\right)\mathcal{E}_{i}^{\alpha}\geq \sum_{i=1}^{N-1}\mathcal{E}_{i}^{\alpha}
\end{aligned}
\end{equation}
for $\alpha\geq \left[1+{\rm log}_{2}(m+2)\right]\gamma\geq 2\gamma>0$, $m\geq0$, where $\tau_{i}=\frac{\mathcal{E}_{1}^{\gamma}+\cdots+\mathcal{E}_{i-1}^{\gamma}}{\mathcal{E}_{i}^{\gamma}}$, $i=2,\cdots,N-1$.
\end{cor}

\bigskip
The Bures measure \cite{VPRK} of entanglement is defined based on distance. For a two-qubit state, the Bures measure $\mathcal{E}_{\mathrm{B}}$ of entanglement as the function of the concurrence $C$ has an analytical formula \cite{SKB}
\begin{align}\label{e:3.2-4}
\mathcal{E}_{\mathrm{B}}(\rho_{AB})=\mathrm{B}(C(\rho_{AB})),
\end{align}
where $B(x)=2-2\sqrt{\frac{1+\sqrt{1-x^{2}}}{2}}$. The authors have shown in \cite{GYG2} that the Bures measure $\mathcal{E}_{\mathrm{B}}$ satisfies the monogamy relation \eqref{e:3-mono1}for $\gamma=1$. Thus, for any $N$-qubit quantum states,
\begin{align}\label{e:3.2-5}
\mathcal{E}^{\alpha}_{\mathrm{B}}(\rho_{A|B_{1}\cdots B_{N-1}})\geq \sum_{i=1}^{N-1}\mathcal{E}^{\alpha}_{\mathrm{B}}(\rho_{AB_{i}})
\end{align}
for any $\alpha\geq 1$. In addition, they generate a monogamy relation as follows
\begin{align}\label{e:3.2-6}
\mathcal{E}^{\alpha}_{\mathrm{B}}(\rho_{A|B_{1}\cdots B_{N-1}})\geq \mathcal{E}^{\alpha}_{\mathrm{B}}(\rho_{AB_{1}})+\sum_{i=2}^{N-1}\left[i^{\alpha}-(i-1)^{\alpha}\right]\mathcal{E}^{\alpha}_{\mathrm{B}}(\rho_{AB_{i}})
\end{align}
for any $\alpha\geq 1$ and $\mathcal{E}^{\alpha}_{\mathrm{B}}(\rho_{AB_{i}})\geq \mathcal{E}^{\alpha}_{\mathrm{B}}(\rho_{AB_{i+1}}),i=1,\cdots,N-2.$

So taking $\mathcal E$ as the Bures measure $\mathcal{E}_{\mathrm{B}}$ in Corollary \ref{e:3.2-cor1} will lead to a
tighter monogamy relation that outperforms those in \cite{GYG1,GYG2}.
\section{Examples}\label{s:example}
In the following, we use the concurrence, convex-roof extended negativity and Bures measure as detailed examples to demonstrate the advantages of our monogamy relations.

Now let us use an example from \cite{TZJF,LSF} to show that our bounds in monogamy relations are the strongest among all studies.

\begin{exmp} Let $\rho=|\psi\rangle\langle\psi|$ be the three-qubit state \cite{AACJLT}:
$$
|\psi\rangle=\lambda_{0}|000\rangle+\lambda_{1}e^{i \varphi}|100\rangle+\lambda_{2}|101\rangle+\lambda_{3}|110\rangle+\lambda_{4}|111\rangle,
$$
where $\sum_{i=0}^4 \lambda_i^2=1$, and $\lambda_{i}\geq 0$ for $i=0,1,2,3,4$.
Then the concurrence $C_{A \mid BC}=2 \lambda_0 \sqrt{\lambda_2^2+\lambda_3^2+\lambda_4^2}$, $C_{AB}=2 \lambda_0 \lambda_2$, and $C_{AC}=2 \lambda_0 \lambda_3$.

\bigskip
The comparison of the monogamy relations for concurrence based on {\bf Case 1}:

$(\rmnum{1})$   Set $\lambda_0=\lambda_1=\lambda_2=\frac{1}{2}, \lambda_3=\lambda_4=\frac{\sqrt{2}}{4}$.
Then we have $C_{A\mid BC}=\frac{\sqrt{2}}{2}, C_{AB }=\frac{1}{2}, C_{AC}=\frac{\sqrt{2}}{4}$. Thus $t=\frac{C_{AC}^{2}}{C_{AB}^{2}}=\frac{1}{2}$. Set $k = 0.8$ and $m=2$ (since $0<t\leq k \leq 1$ and $m\geq 0$).

Therefore, by Theorem \ref{e:3.1-thm2}, for any $\alpha\geq 2\left[1+{\rm log}_{2}(m+2)\right]=6$, our right-hand side of the monogamy relation is
\begin{align*}
Z_1&=C_{AB}^{\alpha}+\left[\frac{(1+k)^{\frac{\alpha}{2}}-1}{k^{\frac{\alpha}{2}}}+k^{\frac{\alpha}{2}}-\frac{C^{\alpha}_{AC}}{C^{\alpha}_{AB}}+\frac{m\alpha}{2}\left(\frac{C_{AB}^{2}}{C_{AC}^{2}}-\frac{1}{k}\right)\right]C_{AC}^{\alpha}\\
&=(\frac{1}{2})^{\alpha}+\left[\frac{(1+0.8)^{\frac{\alpha}{2}}-1}{0.8^{\frac{\alpha}{2}}}+0.8^{\frac{\alpha}{2}}-(\frac{\sqrt{2}}{2})^{\alpha}+\alpha\left(2-\frac{1}{0.8}\right)\right](\frac{\sqrt{2}}{4})^{\alpha}.
\end{align*}

The following right-hand side (RHS) $Z_2$ of the monogamy relation comes from \cite{LSF}, which is a special case of our bound at $m=0$.
\begin{align*}
    Z_2&=C_{AB}^{\alpha}+\left[\frac{(1+k)^{\frac{\alpha}{2}}-1}{k^{\frac{\alpha}{2}}}+k^{\frac{\alpha}{2}}-\frac{C^{\alpha}_{AC}}{C^{\alpha}_{AB}}\right]C_{AC}^{\alpha}\\
&=(\frac{1}{2})^{\alpha}+\left[\frac{(1+0.8)^{\frac{\alpha}{2}}-1}{0.8^{\frac{\alpha}{2}}}+0.8^{\frac{\alpha}{2}}-(\frac{\sqrt{2}}{2})^{\alpha}\right](\frac{\sqrt{2}}{4})^{\alpha}.
\end{align*}

The RHS of the monogamy relation given in \cite{YCFW} is
\begin{align*}
    Z_3&=C_{AB}^{\alpha}+\left[\frac{(1+k)^{\frac{\alpha}{2}}-1}{k^{\frac{\alpha}{2}}}\right]C_{AC}^{\alpha}=(\frac{1}{2})^{\alpha}+\left[\frac{(1+0.8)^{\frac{\alpha}{2}}-1}{0.8^{\frac{\alpha}{2}}}\right](\frac{\sqrt{2}}{4})^{\alpha}.
\end{align*}

The RHS of the monogamy relation given in \cite{JLLF,ZF,JF1}
are respectively:
\begin{align*}
    Z_{_4}=(\frac{1}{2})^{\alpha}+\left(2^{\frac{\alpha}{2}}-1\right)(\frac{\sqrt{2}}{4})^{\alpha},~~
    Z_{_5}= (\frac{1}{2})^{\alpha}+\frac{\alpha}{2}(\frac{\sqrt{2}}{4})^{\alpha},~~
    Z_{_6}= (\frac{1}{2})^{\alpha}+(\frac{\sqrt{2}}{4})^{\alpha}.
\end{align*}

\begin{figure}[H]
\centering
\includegraphics[scale=0.55]{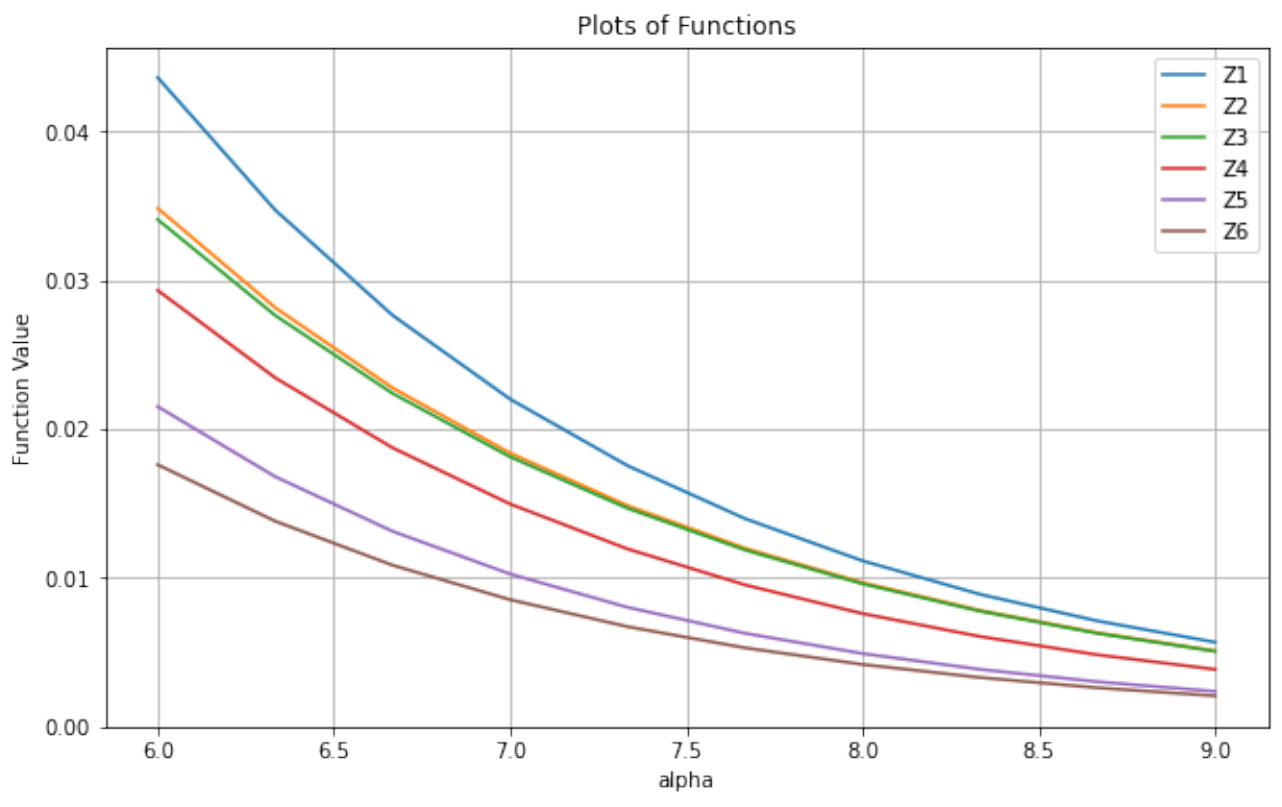}
\end{figure}

Figure 1: The $x$-axis is $\alpha$, and the $y$-axis shows the lower bounds of right hand side of the monogamy relations for the concurrence of $C_{A \mid BC}^\alpha$ (for $\gamma=2$). The blue curve $Z_1$ represents our lower bound at $m=2$ and $k=0.8$. The orange curve $Z_2$, the green curve $Z_3$, the red curve $Z_4$, the purple curve $Z_5$ and brown curve $Z_6$ represent respectively
the lower bounds of the RHS of the monogamy relation from \cite{LSF, YCFW, JLLF,ZF,JF1}.
The graph shows our lower bound $Z_1$ is the strongest. The curves $Z_1-Z_6$ are positioned from top to
bottom respectively.

\bigskip
The comparison of the convex-roof extended negativity based on {\bf Case 1}:

$(\rmnum{2})$ Set $\lambda_0=\lambda_1=\lambda_2=\frac{\sqrt{2}}{3}, \lambda_3=\lambda_4=\frac{\sqrt{6}}{6}$.
Then $\mathcal{N}_{cA \mid BC}=\frac{2\sqrt{10}}{9}$, $\mathcal{N}_{cAB}=\frac{4}{9}$, $\mathcal{N}_{cAC}=\frac{2\sqrt{3}}{9}$. Hence, $t=\frac{\mathcal{N}_{cAC}^{2}}{\mathcal{N}_{cAB}^{2}}=\frac{3}{4}$. Also set $k = 0.8$ and $m=2$. Similar to $(\rmnum{1})$, for any $\alpha\geq6$, our lower bound is
\begin{align*}
W_1&=\mathcal{N}_{cAB}^{\alpha}+\left[\frac{(1+k)^{\frac{\alpha}{2}}-1}{k^{\frac{\alpha}{2}}}+k^{\frac{\alpha}{2}}-\frac{\mathcal{N}_{cAC}^{\alpha}}{\mathcal{N}_{cAB}^{\alpha}}+\frac{m\alpha}{2}\left(\frac{\mathcal{N}_{cAB}^{2}}{\mathcal{N}_{cAC}^{2}}-\frac{1}{k}\right)\right]\mathcal{N}_{cAC}^{\alpha}\\
&=(\frac{4}{9})^{\alpha}+\left[\frac{(1+0.8)^{\frac{\alpha}{2}}-1}{0.8^{\frac{\alpha}{2}}}+0.8^{\frac{\alpha}{2}}-(\frac{\sqrt{3}}{2})^{\alpha}+\alpha\left(\frac{4}{3}-\frac{1}{0.8}\right)\right](\frac{2\sqrt{3}}{9})^{\alpha}.
\end{align*}
The lower bounds of the RHS of the monogamy relation based on \cite{LSF} and \cite{YCFW} are $W_2$ and $W_3$ respectively, where
\begin{align*}
    W_2&=(\frac{4}{9})^{\alpha}+\left[\frac{(1+0.8)^{\frac{\alpha}{2}}-1}{0.8^{\frac{\alpha}{2}}}+0.8^{\frac{\alpha}{2}}-(\frac{\sqrt{3}}{2})^{\alpha}\right](\frac{2\sqrt{3}}{9})^{\alpha},\\
     W_3&=(\frac{4}{9})^{\alpha}+\left[\frac{(1+0.8)^{\frac{\alpha}{2}}-1}{0.8^{\frac{\alpha}{2}}}\right](\frac{2\sqrt{3}}{9})^{\alpha}.
\end{align*}
The RHS of the monogamy relation given in \cite{JLLF,ZF,JF1} are
\begin{align*}
    W_{_4}=(\frac{4}{9})^{\alpha}+\left(2^{\frac{\alpha}{2}}-1\right)(\frac{2\sqrt{3}}{9})^{\alpha},~~
    W_{_5}= (\frac{4}{9})^{\alpha}+\frac{\alpha}{2}(\frac{2\sqrt{3}}{9})^{\alpha},~~
    W_{_6}= (\frac{4}{9})^{\alpha}+(\frac{2\sqrt{3}}{9})^{\alpha}.
\end{align*}
\begin{figure}[H]
\centering
\includegraphics[scale=0.55]{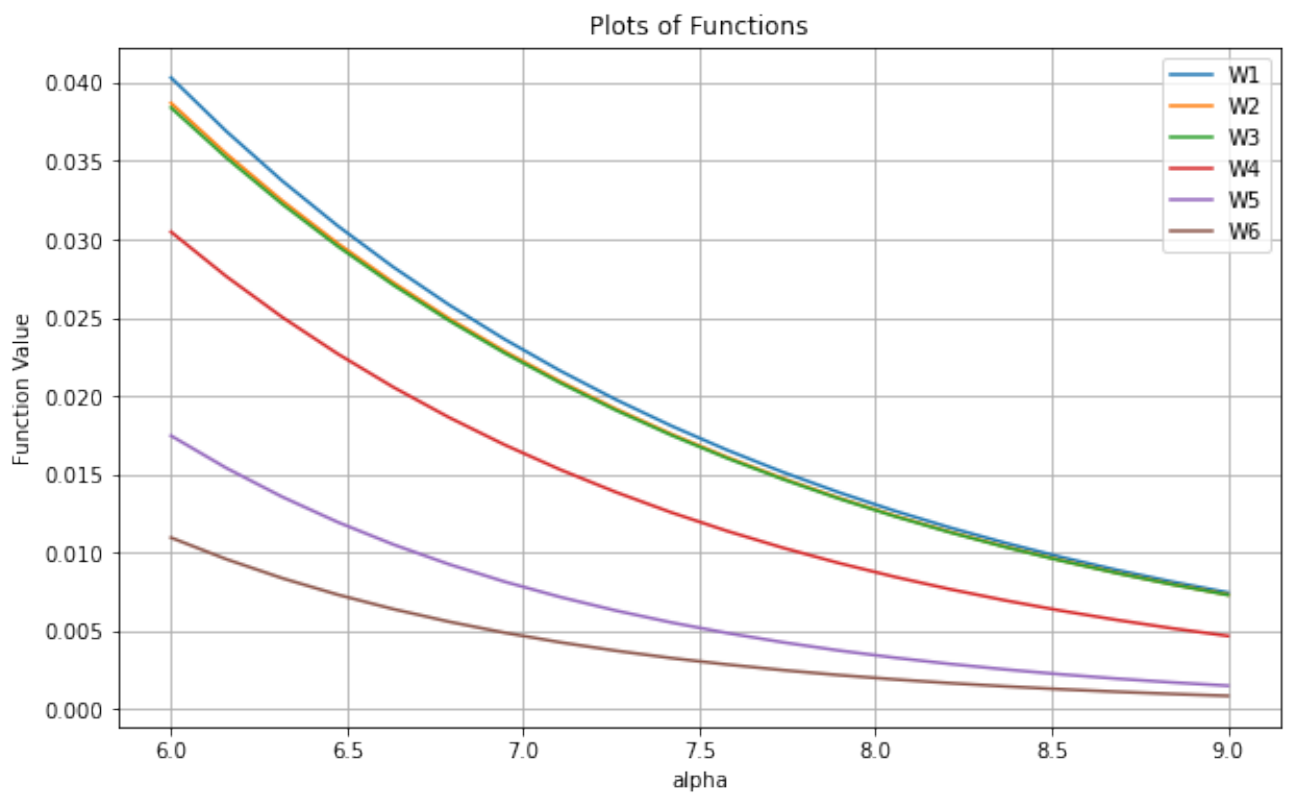}
\end{figure}
Figure 2: The $x$-axis is $\alpha$, and the $y$-axis shows the lower bounds of the concurrence measure of $\mathcal{N}_{cA \mid BC}$ (for $\gamma=2$). The blue curve $W_1$ represents our lower bound at $m=2$ and $k=0.8$. The orange curve $W_2$ represents the lower bound from \cite{LSF}. The green curve $W_3$ represents the lower bound from \cite{YCFW}. The red curve $W_4$, purple  curve $W_5$, and brown curve $W_6$ respectively represent the lower bound from \cite{JLLF, ZF,JF1}. The graph shows our lower bound $W_1$ is the strongest. The curves $W_1-W_6$ are positioned from top to
bottom respectively.

\bigskip
The comparison of the Bures measure based on {\bf Case 2}:

$(\rmnum{3})$ Set $\lambda_0=\lambda_3=\frac{\sqrt{2}}{3}, \lambda_1=\lambda_4=0, \lambda_2=\frac{\sqrt{5}}{3}$.
Then by \eqref{e:3.2-4} we have $\mathcal{E}_{\mathrm{B}}(\rho_{AB})=\mathrm{B}(C_{AB})\approx 0.14989$, $\mathcal{E}_{\mathrm{B}}(\rho_{AC})=\mathrm{B}(C_{AC})\approx 0.05279$, thus, $\tau_{2}=\frac{\mathcal{E}_{\mathrm{B}}(\rho_{AB})}{\mathcal{E}_{\mathrm{B}}(\rho_{AC})}\approx2.83936$ $(\gamma=1)$. Suppose $m=2$, by Corollary \ref{e:3.2-cor1}, for any $\alpha\geq \left[1+{\rm log}_{2}(m+2)\right]=3$, our lower bound is
\begin{align*}
T_1&=\mathcal{E}^{\alpha}_{\mathrm{B}}(\rho_{AB})+(2^{\alpha}-\tau_{2}^{-\alpha})\mathcal{E}^{\alpha}_{\mathrm{B}}(\rho_{AC})+m\alpha(\tau_{2}-1)\mathcal{E}^{\alpha}_{\mathrm{B}}(\rho_{AC})\\
&\approx 0.14989^{\alpha}+(2^{\alpha}-2.83936^{-\alpha})0.05279^{\alpha}+2\alpha(2.83936-1)0.05279^{\alpha}.
\end{align*}

The following  for RHS of the monogamy relation is based on \eqref{e:2-2},
\begin{align*}
T_2&=\mathcal{E}^{\alpha}_{\mathrm{B}}(\rho_{AB})+(2^{\alpha}-\tau_{2}^{-\alpha})\mathcal{E}^{\alpha}_{\mathrm{B}}(\rho_{AC})\approx 0.14989^{\alpha}+(2^{\alpha}-2.83936^{-\alpha})0.05279^{\alpha}.
\end{align*}

The RHS of the monogamy relation based on \cite{GYG1,GYG2} is
\begin{align*}
T_3&=\mathcal{E}^{\alpha}_{\mathrm{B}}(\rho_{AB})+(2^{\alpha}-1)\mathcal{E}^{\alpha}_{\mathrm{B}}(\rho_{AC})\approx 0.14989^{\alpha}+(2^{\alpha}-1)0.05279^{\alpha}.
\end{align*}

The RHS of the monogamy relation given in \cite{GYG2} is
\begin{align*}
T_4&=\mathcal{E}^{\alpha}_{\mathrm{B}}(\rho_{AB})+\mathcal{E}^{\alpha}_{\mathrm{B}}(\rho_{AC})\approx 0.14989^{\alpha}+0.05279^{\alpha}.
\end{align*}
\begin{figure}[H]
\centering
\includegraphics[scale=0.55]{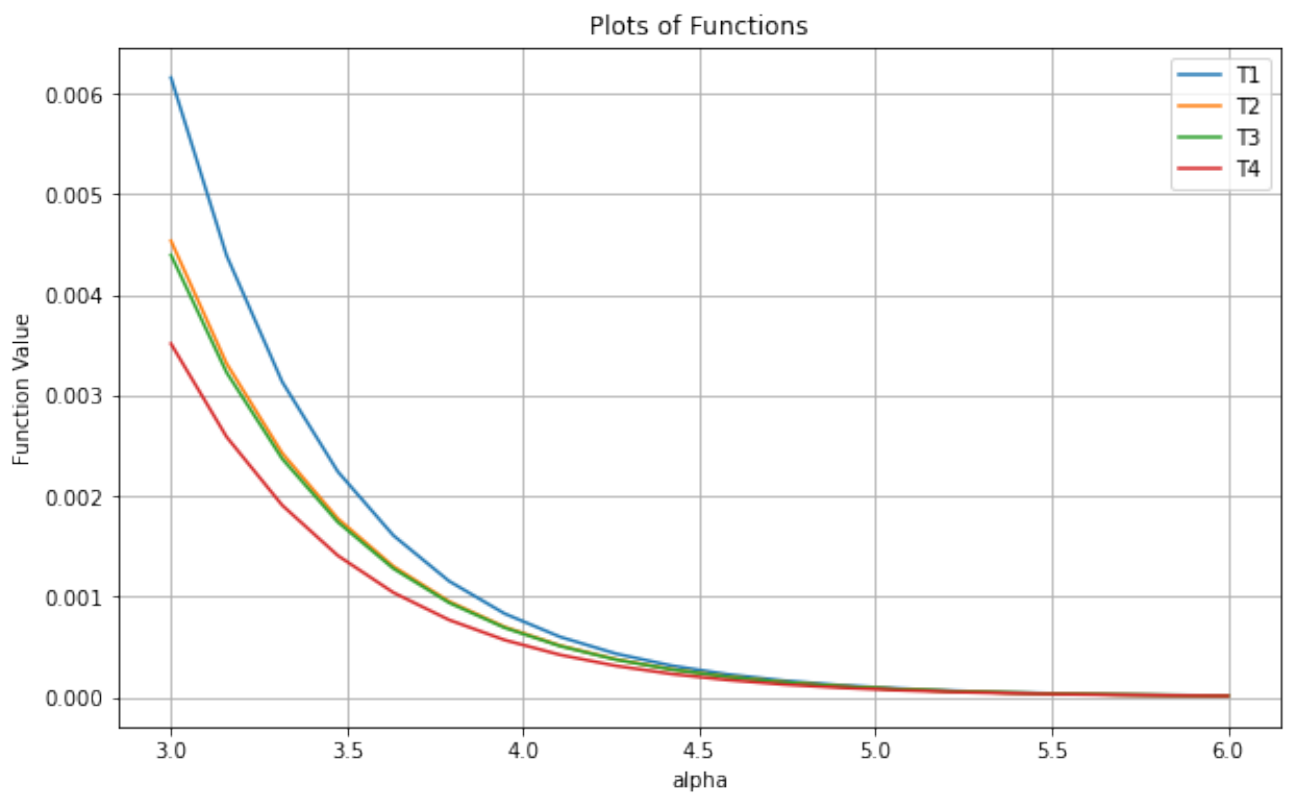}
\end{figure}
Figure 3: The $x$-axis is $\alpha$, and the $y$-axis shows the lower bounds of the Bures measure of $\mathcal{E}_{\mathrm{B}}(\rho_{A \mid BC})$ (for $\gamma=1$). The blue curve $T_1$ represents our lower bound at $m=2$. The orange curve $T_2$ represents the lower bound from $(2.2)$. The green curve $T_3$ and red curve $T_4$ respectively represent the lower bound from \cite{GYG1, GYG2}. The graph shows our lower bound $T_1$ is the strongest.
The curves $T_1-T_4$ are positioned from top to
bottom respectively.
\end{exmp}

\section{\textbf{Conclusion}}
By introducing a family of inequalities indexed by $m$,
 we get a family of monogamous relations in a unified manner. We illustrate in detail how the parameter helps us to obtain stronger
 monogamy relations in various cases. It turns out that by choosing the region
 $[(1+\log_2(m+2))\gamma, \infty)$ for $\alpha$,  where
 $m=\lfloor 2^{\frac{\alpha}{\gamma}-1}-2\rfloor$, 
 one can get stronger $\alpha$th-monogamy relations. Comparisons are given in detailed examples which show that our bounds for RHS of the monogamy relation are indeed the {\it tightest} in both cases.

Our study is conducted for the monogamy relation based on any quantum correlation measure, thus the improved generalized monogamy relations hold for quantum coherence.

\bibliographystyle{plain}

\end{document}